\documentclass[12pt,conference,onecolumn]{IEEEtran}

\usepackage{amsmath}
\usepackage{amsfonts}
\usepackage{amssymb}
\usepackage{tabularx}{\tiny }
\usepackage{upref}
\usepackage{theorem}

\usepackage{graphics}
\usepackage[pdf_tex]{graphicx} 

\usepackage{psfrag}
\usepackage{enumerate}
\usepackage{caption}
\usepackage{subcaption}
\usepackage{color}
\usepackage{booktabs}
\usepackage{multirow}
\usepackage{multicol}
\usepackage{diagbox}
\usepackage{pict2e}
\usepackage{longtable}
\usepackage{bigdelim}
\usepackage[font=small,labelfont=bf,tableposition=top]{caption}

\newtheorem{theorem}{Theorem}
\newtheorem{lemma}{Lemma}

\newtheorem{corollary}{Corollary}
\newtheorem{definition}{Definition}
\newtheorem{remark}{Remark}
\newtheorem{example}{Example}
\newtheorem{claim}{Claim}

\newcommand{\comment}[1]{}

\def\cS{\mbox{$\cal{S}$}}

\def\cD{\mbox{$\cal{D}$}}

\newcommand{\TRS}{X{}} 

\newcommand{\ND}[1]{{\overline{D}_{\tilde{#1}}}}
\newcommand{\Nd}[1]{{\overline{d}_{\tilde{#1}}}}

\newcommand{\NS}[1]{{\overline{S}_{\tilde{#1}}}}
\newcommand{\Ns}[1]{{\overline{s}_{\tilde{#1}}}}

\newcommand{\Rt}{R^p(N,K,M)} 
\newcommand{\Rto}{R^{*p}(N,K,M)} 
\newcommand{\Rm}{R^{*}(N,K,M)} 
\newcommand{\Rtc}{R^{p}_{c}(N,K,M)}

\title{Demand Private Coded Caching}
\author{
	\IEEEauthorblockN{Sneha~Kamath}
	\IEEEauthorblockA{ Qualcomm, India. Email: snehkama@qti.qualcomm.com}
	\thanks{This work was done while the author was at IIT Bombay. She acknowledges fruitful discussions with Jithin Ravi and Bikash Kumar Dey.}
	} 
  
\begin{document}
  	 \IEEEoverridecommandlockouts
 \maketitle
 \begin{abstract} 
 	The work by Maddah-Ali and Niesen demonstrated the benefits in reducing the transmission rate in a noiseless broadcast network by joint design of caching and delivery schemes. In their setup, each user learns the demands of all other users in the delivery phase. In this paper, we introduce the problem of {\em demand private coded caching} where we impose a privacy requirement that 
 	no user learns any information about the demands of other users. We provide an achievable scheme and compare its performance using the existing lower bounds on the achievable rates under no privacy setting. For this setting, when $N\leq K$ we show that our scheme is order optimal within a multiplicative factor of 8. Furthermore, when $N > K$ and $M\geq N/K$, our scheme is order optimal within a multiplicative factor of 4.
 \end{abstract}

 \section{Introduction}
 To reduce the network traffic at peak hours, some of the popular contents are stored at users' end during off-peak hours.
Such techniques are called coded caching~\cite{Maddah16}, and it is studied from an information theoretic perspective recently.
Coded caching was studied in~\cite{Maddah14} for a broadcast network with one server and many users. In their setting, the server has $N$ files $W_i, i \in [N]:= \{1,\ldots,N\}$, and it can broadcast to $K$ users  through a noiseless broadcast link. 
Each user can store  $M$ out of $N$ files, where $M$ varies from 0 to $N$.
Each user requests one of the  $N$ files in the delivery phase and the demand of each user is conveyed to the server. 
The server broadcasts a message to meet the demands of all the users. 
For a set of demands, the server chooses an encoding function.
This particular choice of of function is conveyed in the broadcast phase to all users.
Each user's decoding function is chosen based on the encoding function.
 Thus, each user learns the demands of all other users. 
This violates the privacy of users. So in this paper, we study the coded caching problem with the additional privacy constraint that each user should not learn any information about the demands of other users. For this setup, we address the trade-off between transmission rate and cache size. We call this as {\em demand private coded caching}.

The phase in which the server stores a fraction of all the files in the caches of users is called as placement phase, and the phase in which the server broadcasts is called as delivery phase~\cite{Maddah14}. Each user decodes the demanded file using the cache content and the broadcast message received. It is assumed that all the files are of equal size and the transmission rate is the number of bits transmitted per size of one file. For any $N$ and $K$, an achievable scheme is provided for each cache 
size $M$, and its order optimality is shown to be within a multiplicative factor of 12. Improved achievable rates were obtained  in~\cite{Amiri17, Zhang18, Vilardebo18, Yu18} and the lower bounds on the achievable rates were improved in~\cite{Ghasemi17,Wang18}.
The optimal rates when the cache content is not allowed to be coded were characterized in~\cite{Yu18, Wan16}.

Security aspects in coded caching have been considered under information theoretic security in some recent works~\cite{Sengupta15}. Security against a wiretapper who observes the broadcasted message has been studied in~\cite{Sengupta15}. In this work, the security is obtained by distributing {\em keys} among  users. They provided an approximate characterization of memory-rate trade-off for this problem. 
Privacy aspects have also been studied in coded caching~\cite{Ravindrakumar18} in the following sense. Each user should not obtain any information about any other file than the requested one. A feasible scheme has been proposed by distributing keys in the placement phase, and the order optimality of this scheme is shown to be within a constant factor~\cite{Ravindrakumar18}. 

In this paper, we study the coded caching setup introduced by Maddah-Ali and Niesen~\cite{Maddah14} under the constraint that no user learns any information about the demands of the other users.  For $N$ files and $K$ users,  we obtain an achievable scheme from a coded caching scheme for $N$ files $NK$ users with no privacy requirement. The  memory-rate pair achieved by our scheme is given in Theorem~\ref{Thm_genach}.
Since the coded caching with demand privacy is a more constrained problem compared to the problem without the privacy requirement, any converse for the problem under no privacy also holds as a converse for the problem with the privacy requirement. By comparing our achievable memory-rate pair with the lower bounds on the achievable   memory-rate pair under no privacy, in Theorem~\ref{Thm_tightness}, we show that our scheme is order optimal within a 
multiplicative factor of 8 when $N\leq K$. We also show that when $N > K$, our scheme is order optimal within a multiplicative factor of 4 if $M \geq N/K$.


Recently, an idependent and parallel work~\cite{Wan19} on demand privacy for coded caching was posted on arXiv on 28 August 2019. The problem formulation of~\cite{Wan19} is very similar to ours.
 We studied the problem of demand privacy for coded caching for the case of single request from users, the setup studied in~\cite{Maddah14}.
The authors in~\cite{Wan19} studied the cases of single request as well as multiple requests from the users. So it is good to compare our results with theirs for the single request case.
For the single request case, the achievable schemes in these two works are very different. 
In both of these works, the  tightness of the achievable rates are shown  by comparing it with the existing converse results on coded caching problem under no privacy requirement. 
The general scheme in~\cite{Wan19} is shown to be order optimal when $M\geq N/2$. 
In contrast, we show that our scheme is order optimal within a factor of 8 when $N \leq K$.
A special case where all users demand distinct files was also studied in~\cite{Wan19}. This assumption implies that $N\geq K$. They have provided an improved scheme for this special case which is order optimal within a factor of $4$ when $M\geq N/K$. We show that when $N>  K$, our scheme is order optimal within a factor of $4$ when $M \geq N/K$ without any restriction on the demand vectors.

The paper is organized as follows. We give our problem formulation in Section~\ref{sec_problem}. Our results are presented in Section~\ref{sec_results}, and the proofs are given in Section~\ref{sec_proofs}. We conclude the paper  in Section~\ref{sec_discuss} with some simulation results.

 \section{Problem formulation and definitions}
 \label{sec_problem}

A server with $N$ files is connected to $K$ users via a noiseless broadcast link. 
We assume that $N$ files $W_1, W_2,\ldots,W_N$ are independent, and each file $W_i$ is of length $F$ bits and takes values uniformly in the set $[2^F] := \{1,\ldots, 2^F\}$. 
User $k$ has a local cache $Z_k$ 
of size $MF$ bits for some real number $M \in [0,N]$. In the placement phase, the server fills the cache of each user. In the delivery phase,
each user $k$ demands file $W_{D_k}$ from the server, where  $D_{k} \sim \text{Unif}\{[N]\}  \; \forall   k \in [K] $.   Let  $\overline{D}=\{D_{1},D_{2},...,D_{K}\}$ denote the demands of all users, and let  $\ND{k}$ denote all demands, but $D_k$, i.e., $\ND{k}= \overline {D} \setminus \{ D_{k}\}$.  All users convey their demands to the server without revealing it to other users. Then, the server broadcasts a message of size $RF$ bits to all the users. Here $R$ is a real number, and is defined as the rate of transmission. Using the received message and the cache content, each user $k$ recovers file $W_{D_k}$.

In addition to the recovery requirement, we have a privacy requirement that user $k$ should not gain any information about $\ND{k}$. 
To achieve this, we consider some {\em shared randomness} $S_k$ which is shared between user $k$ and the server, and  it is not known to the other users.
The shared randomness can be achieved during the placement phase since the placement is done secretly for each user. Random variables $S_1,\ldots, S_K$ take values in some finite alphabets $\cS_1,\ldots, \cS_K$, respectively.
The set of random variables $(S_1, \ldots, S_{K})$ is denoted by $\overline{S}$.
The random variables $ \{S_k: 1\leq k \leq K \} \cup  \{D_k: 1\leq k \leq K \} \cup \{W_i: 1\leq i \leq N \}$  are independent of each other.

\begin{figure}[htb]
  \centering
   \includegraphics[scale=0.5]{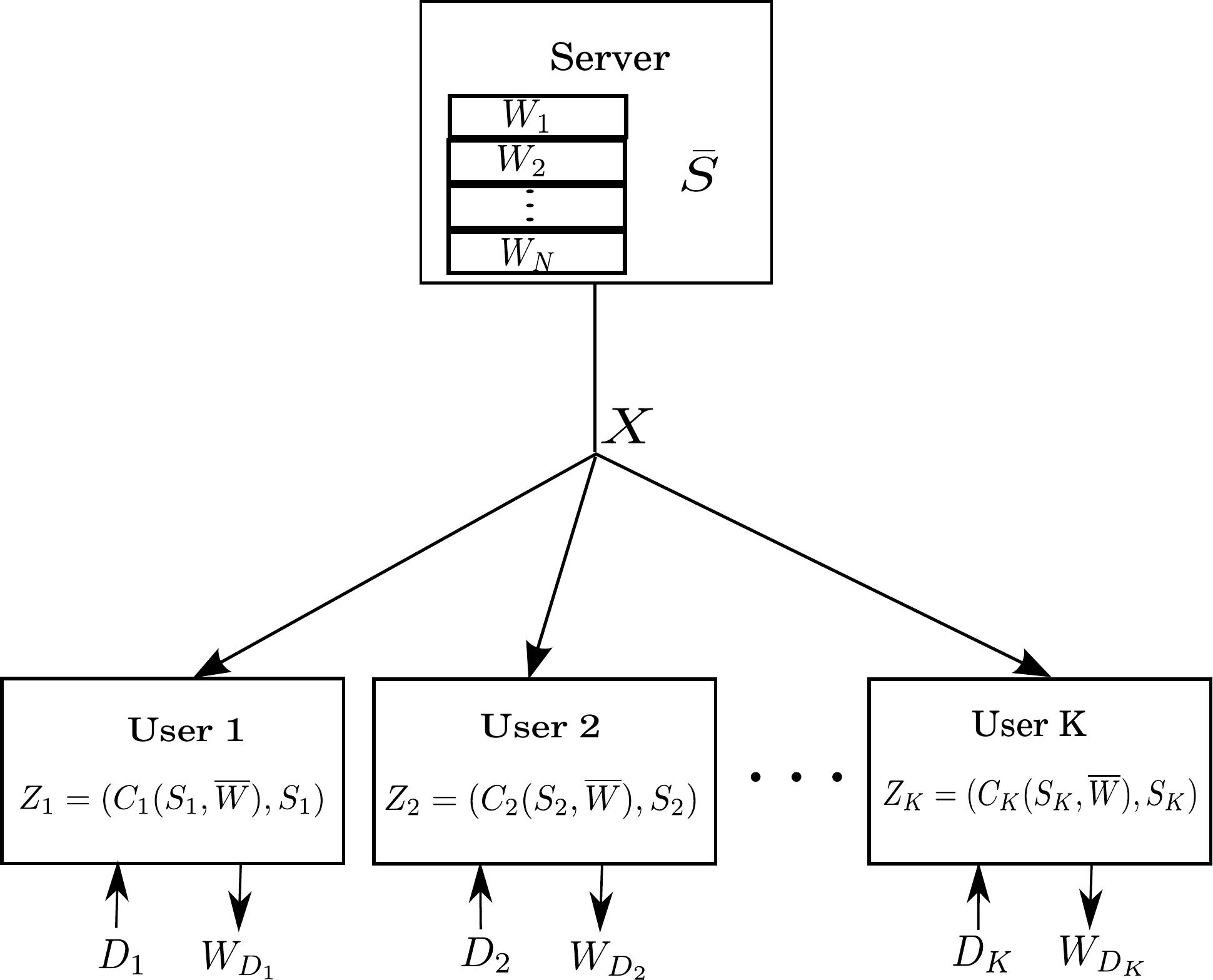}
  \caption{Coded caching model}
\end{figure}

A caching scheme consists of $K$  cache encoding functions, the transmission encoding function, and $K$ decoding functions.
For $k \in [K]$, the cache encoding function for the $k$-th user is given by
\begin{align}
C_{k} :\cS_k \times {[2^F]}^N  \rightarrow [2^{MF}], \label{Def_cach_enc}
\end{align} 
and the cache content $Z_k$ is given by $Z_k =(C_k(S_k, \overline{W}), S_k)$. The   transmission
encoding function
\begin{align}
E: {[2^F]}^N \times \cD_1 \cdots \times \cD_N \times \cS_1 \times \cdots \times \cS_K \rightarrow [2^{RF}], \label{Def_Tx_enc}
\end{align}
outputs the transmitted message $X=E(\overline{W}, \overline{D}, \overline{S})$. And finally, user $k$ has decoding function
\begin{align}
G_{k} : \cD_k \times \cS_k \times 2^{RF} \times 2^{MF}  \rightarrow [2^{F}], \label{Def_dec}
\end{align}
which recovers user $k$'s demanded file $W_{D_{k}}$.
A memory-rate pair $(M,R)$ is said to be {\em achievable} with demand privacy if
a scheme satisfies  the following two conditions:
\begin{align} 
H\left(W_{D_{k}}|Z_k,\TRS, D_k\right)  = 0 \quad \forall k=1,\ldots,K, \label{Eq_dec_cond}
\end{align}
and
\begin{align}
 & I\left(\ND{k};Z_k,\TRS, D_k \right)  = 0 \quad \forall   k=1,\ldots,K. \label{Eq_instant_priv}
\end{align}
For $N$ files and $K$ users and memory $M$, memory-rate trade-off is defined as 
\begin{align}
R^{*p}(N,K,M)&=\inf\{R:(M,R) \mbox{ is achievable with demand privacy}\}. \label{Eq_opt_rate_priv}
\end{align}

A memory-rate pair $(M,R)$ is achievable under no  privacy if it satisfies~\eqref{Eq_dec_cond}. In this case, the memory-rate trade-off is defined as
\begin{align}
R^{*}(N,K,M) &=\inf\{R:(M,R) \mbox{ is achievable} \}.\label{Eq_opt_rate_nopriv}
\end{align}

 \section{Main results}
 \label{sec_results}
%
%
%
%


First let us consider a trivial scheme which achieves demand privacy. In the caching phase, caches of all users are filled with the same $M$ out of $N$ files. Then in the broadcast phase, 
the rest of $N - M$ files are transmitted. In this scheme, all users can obtain all files, and the rate of transmission is $N-M$. It is easy to see that this scheme satisfies~\eqref{Eq_dec_cond} and~\eqref{Eq_instant_priv}. Next we give a coded caching instance for which we give a scheme with rate strictly less than $N-M$ and it satisfies~\eqref{Eq_dec_cond} and~\eqref{Eq_instant_priv}. Let us consider a simple example with	 $N=K=2$ and $M=1$.
\begin{example}
	\label{Ex_simple}
Let $A$ and $B$ denote the files $W_1$ and $W_2$, respectivley. For $F=6$, next we show a scheme which achieves demand privacy~\eqref{Eq_instant_priv} . Let us denote the 6 bits of $A$ and $B$ by $A_1,A_2,A_3,A_4,A_5,A_6$ and $B_1,B_2,B_3,B_4,B_5,B_6$, respectively. Each user stores 3 bits of each file. In the caching phase, there are two choices for the cache content for each user. 
Let $S_i,i=1,2$ be  uniformly distributed in $\{1,2\}$. Then the cache content of user $i$, denoted by $Z_{iS_i}$, is given in  Table~\ref{Table_cache_NK2}. 
The realization of $S_i,i=1,2$ is known only to user $i$ and the server.

After knowing the demand requests, in the delivery phase, the server chooses one of the 4 transmissions $T_i, i=1,2,3,4$ given in Table~\ref{Tab_Tx_NK2}.
The choice is determined by the random variables $S_1,S_2$ and $D_1,D_2$ and it is given in Table~\ref{Tab_Tx_NK2}. For example, if $S_1= 1, S_{2}=2$ and $D_1=1, D_2=1$, then the server chooses transmission $T_2$. Using Tables~\ref{Table_cache_NK2} and~\ref{Tab_Tx_NK2}, it is easy to verify that this scheme satisfies the decoding condition~\eqref{Eq_dec_cond}. Now let us see how this scheme achieves demand privacy.  Let us consider the case of $S_1= 1, S_{2}=2$ and $D_1=1, D_2=1$. Then from the point of view of user 1, with transmission $T_2$ user 2 may be decoding file $A$ using cache $Z_{22}$ or he may be decoding file $B$ using cache $Z_{21}$.
User 2 also has similar inference. As we can observe that not revealing the caching function along with  the clever choice of transmission preserves the demand privacy against the other user. In this example, each of the transmission consists of 4 bits. Since $F=6$, this scheme achieves a rate $R = \frac{4}{6} = \frac{2}{3}$. In contrast, the optimal rate wihout any privacy requirement for this example is shown~\cite{Maddah14} to be $\frac{1}{2}$.
\end{example}


 \begin{table}[h]
	\centering
	\begin{tabular}{|c|c|c|c|c|c|c|}
		\hline
		Cache & \multicolumn{6}{c|}{Cache content}\\
		\hline
		$Z_{11}$ & $A_1$ & $A_2$ & $A_3$ & $B_1$ & $B_2$ & $B_3$ \\
		$Z_{12}$ & $A_1$ & $A_4$ & $A_5$ & $B_1$ & $B_4$ & $B_5$ \\
		$Z_{21}$ & $A_2$ & $A_4$ & $A_6$& $B_2$&$ B_4$& $B_6$ \\
		$Z_{22}$ & $A_3$ & $ A_5 $& $A_6$ &$ B_3 $& $B_5 $&$ B_6$ \\
		\hline
	\end{tabular}
	\caption{Choices for the caches of user 1 and 2}
	\label{Table_cache_NK2}
\end{table}
 
 \begin{table}[h]
	\centering
	\begin{tabular}{|c|c|c|c|c|c|}
		\hline
		$S_1$ & $S_2$ & $D_1$ & $D_2$ & Transmission index & Transmissions\\
		\hline
		$1$ & $1$ & $1$ & $1$ &   & $ B_2 \oplus A_1 \oplus A_4 $ \\
		$1$ & $2$ & $1$ & $2$ & $T_1$  & $B_4 \oplus B_6 \oplus A_5$ \\
		$2$ & $1$ & $2$ & $1$ &   & $ B_2 \oplus A_6 \oplus A_3$\\
		$2$ & $2$ & $2$ & $2$ &   & $ B_1 \oplus A_5 \oplus B_3$\\
		\hline
		$1$ & $1$ &$1$ & $2$&   & $B_1 \oplus A_4 \oplus B_2 $ \\
		$1$ & $2$ &$1$ & $1$& $T_2$ & $A_4 \oplus B_6 \oplus B_5 $\\
		$2$ & $1$ &$2$ & $2$&  & $A_2 \oplus A_6 \oplus B_3 $\\
		$2$ & $2$ &$2$ & $1$&  & $A_1 \oplus A_5 \oplus B_3 $\\
		\hline
		$1$ & $1$ &$2$ & $1$ &  & $B_4 \oplus A_2 \oplus A_1 $ \\
		$1$ & $2$ &$2$ & $2$ & $T_3$ & $A_6 \oplus A_5 \oplus B_4 $\\
		$2$ & $1$ &$1$ & $1$ &  & $B_6 \oplus A_3 \oplus B_2 $\\
		$2$ & $2$ &$1$ & $2$&  & $B_5 \oplus A_3 \oplus B_1 $\\
		\hline
		$1$ & $1$ &$2$ & $2$&   & $B_4 \oplus A_2 \oplus B_1 $ \\
		$1$ & $2$ &$2$ & $1$& $T_4$ & $A_6 \oplus B_5 \oplus A_4 $\\
		$2$ & $1$ &$1$ & $2$&  & $B_6 \oplus B_3 \oplus A_2 $\\
		$2$ & $2$ &$1$ & $1$&  & $B_5 \oplus A_3 \oplus A_1 $\\
		
		\hline
	\end{tabular}
	\caption{ Transmission design}
	\label{Tab_Tx_NK2}
\end{table}

 The scheme described in Example~\ref{Ex_simple} for $N=K=2$ can be generalized for any arbitrary $N$ and $K$.
Our scheme achieves the memory-rate pair given in the following theorem. The achievable memory-rate pair in the following theorem is the achievable memory-rate pair using the scheme given in~\cite{Yu18} for $N$ files and $NK$ users and memory $M$.
 \begin{theorem}
 	\label{Thm_genach}
 	For  $N$ files,  $K$ users and memory $M$, the following holds:
 	\begin{align}
 		R^{*p}(N,K,M) \leq 
 	 	\Rt := \frac{{NK \choose KM+1 }-{NK-N \choose KM+1}}{{NK \choose KM}}, \textnormal{ if } M \in \{0, 1/K,2/K,\ldots,N\}.
 		\label{Eq_achvble_regn}
 	\end{align}
 \end{theorem}

\begin{definition}
	\label{Def_convx}
Let $\Rtc$ denote the lower convex envelope of the points given in~\eqref{Eq_achvble_regn}. 
\end{definition}

We have the following corollary.
\begin{corollary}
	\label{Coro_tight}
The rate region given by $\Rtc$ is achievable.
\end{corollary}

\par 
\comment{
The proposed scheme in Theorem~\ref{Thm_genach} which achieves $\Rt$ in \eqref{Eq_achvble_regn} uses the coded caching scheme for $N$ files and $NK$ users under no privacy.
Satisfying the subset of $N^K$ demand vectors stated in the theorem will achieve privacy whereas the rate achieved is the peak rate over the the whole set of $N^{NK}$
demand vectors. We can improve this scheme by designing cache and transmission to satisfy only the demand subset in Theorem~\ref{Thm_genach} and use coded prefetching.
}

\comment{
For $M=0$,  we cannot have any cache randmoness. So from Theorem~\ref{Thm_no_random}, we get $R^{*p}(0) = N$. However, under no privacy the optimal rate $R^{*}(0) = \min (N,K)$. So  for $M=0$ with $N \geq K$,
the optimal transmission rate with demand privacy differs by a multiplicative factor of $N/K$ from that of no  privacy. 
}

In Theorem~\ref{Thm_tightness}, we discuss the tightness of the achievable rate region given in Corollary~\ref{Coro_tight}. 

\begin{theorem} 
	\label{Thm_tightness} 
 $\Rtc$, defined in Definition~\ref{Def_convx}, satisfies the following:
	\begin{enumerate}
		\item For $N \leq K$, it is within a multiplicative gap of 8 from the optimal region, i.e.,  $\frac{\Rtc}{\Rto} \leq 8$. \label{Thm_tight_part1}
		\item For $N > K$, if $M \geq N/K$, then it is within a multiplicative gap of 4 from the optimal region, i.e.,  $\frac{\Rtc}{\Rto} \leq 4$. \label{Thm_tight_part2}
		\item For all $N$ and $K$,  it is optimal  for $M \geq \frac{(NK-1)}{K}$,  i.e.,  $\Rtc = \Rto $. \label{Thm_tight_part3}
	\end{enumerate}
\end{theorem}

%
 \section{Proofs of the results}
  \label{sec_proofs}

\subsection{Proof of Theorem~\ref{Thm_genach}}

In the following,  coded caching schemes for $N$ files and $K$ users and memory $M$ with and without demand privacy are called  as $(N,K,M)$-private scheme and $(N,K,M)$-non-private scheme, respectively.
We derive an $(N,K,M)$-private scheme from an $(N,NK,M)$-non-private scheme. In particular, we use the scheme given in~\cite{Yu18} as the  $(N,NK,M)$-non-private scheme to prove the theorem.
For the $(N,NK,M)$-non-private scheme, let $C_k^{\text{np}}$ and $G_k^{\text{np}}$ denote user $k$'s cache encoding function and the decoding function, respectively, where $k=1,\ldots, NK$, and let $E^{\text{np}}$ denote the transmission encoding function. The definitions of  $C_k^{\text{np}}, E^{\text{np}}$  and $G_k^{\text{np}}$ are similar to the definitions given in~\eqref{Def_cach_enc}, ~\eqref{Def_Tx_enc} and ~\eqref{Def_dec}, respectively without any shared randomness.

 First let us consider the procedure of choosing the caches in  $(N,K,M)$-private scheme from  the $(N,NK,M)$-non-private scheme. Let the shared key $S_k$ be distributed as
$S_k \sim \text{Unif}\{[N]\}$. 
Then $k$-th user's cache encoding function $C_k,k=1,\ldots,K$ in $(N,K,M)$-private scheme is given by $C_k = C^{\text{np}}_{(k-1)N+s_k}$ for $S_k = s_k$.
 
 Now we explain the procedure in choosing $E$ from $E^{\text{np}}$. Let $\bar{d}$ and $\bar{d}^{\text{np}}$ denote a demand vector in $(N,K,M)$-private scheme and in the $(N,NK,M)$-non-private scheme, respectively.
 The transmission encoding function $E^{\text{np}}$ in  the $(N,NK,M)$-non-private scheme is determined by the demand vector $\bar{d}^{\text{np}}$. Whereas, in $(N,K,M)$-private scheme, it is determined by the demand vector $\overline{d}$ and the shared keys $\overline{s}$.
 So, we denote $E$ and $E^{\text{np}}$ by $E_{(\bar{d},\bar{s})}$ and $E^{\text{np}}_{\bar{d}^{\text{np}}}$, respectively.
   For given $\overline{S} = \overline{s}$ and $\overline{D} = \overline{d}$,  
    in  $(N,K,M)$-private scheme, the server chooses a $\bar{d}^{\text{np}}$ from
   the $(N,NK,M)$-non-private scheme as follows:
   
  For given $\overline{S} = \overline{s}$ and $\overline{D} = \overline{d}$, let 
  \begin{align}
  c_k := (s_k - d_k)  \mod N \mbox{ for } k = 1, \ldots, K. \label{Def_C}
  \end{align}
  Further, let
  \begin{align*}
  \bar{c}:= (c_1, \ldots, c_K).
  \end{align*}
Now consider the $N$-length vector $q_k, k=1, \ldots, K$ in $[N]^{N}$ obtained by applying $c_k$ right cyclic shifts on $( 1, 2, \ldots, N)$. Then $\bar{d}^{\text{np}}$ is given by $\bar{d}^{\text{np}}= (\overline{q_1},\overline{q_2},\ldots, \overline{q_{K}})$. Here, we can observe that a $\bar{c}$ uniquely determines a $\bar{d}^{\text{np}}$. The encoding function $E_{(\bar{d},\bar{s})}(\overline{w})$  in $(N,K,M)$-private scheme is chosen as $E_{(\bar{d},\bar{s})}(\overline{w}) = E^{\text{np}}_{\bar{c}}(\overline{w})$.

In the $(N,NK,M)$-non-private scheme, the server also communicates the demand vector in the transmission to enable the decoding at each user. Thus, the broadcast message contains the demand vector and a function of the files. For demand vector $ \overline{d}^{\text{np}}$, let $X'_{\overline{d}^{\text{np}}}:= E^{\text{np}}_{\overline{d}^{\text{np}}}(\overline{W})$ in the $(N,NK,M)$-non-private scheme. Thus, the transmitted message $X^{\text{np}}$ in the $(N,NK,M)$-non-private scheme
is given by $X^{\text{np}} = (\overline{D}^{\text{np}}, X^{\text{np}}_{\overline{D}^{\text{np}}})$.
In $(N,K,M)$-private scheme, the transmitted message $X$ is given by  $X:= (\bar{C},X^{\text{np}}_{\bar{C}})$.
%


The decoding at $k$-th user in $(N,K,M)$-private scheme is as follows.
Since $\bar{c}$ is broadcasted, $\bar{d}^{\text{np}}$ 
for which $X^{\text{np}}_{\bar{c}}$ happens in the $(N,NK,M)$-non-private scheme
is known to the user. Then, the user follows the decoding of $(k-1)N+s_k$-th user in the $(N,NK,M)$-non-private scheme because the user has the cache of $(k-1)N+s_k$-th user in the $(N,NK,M)$-non-private scheme. 
For a given $\bar{c}$, the $i$-th user's demand $\bar{d}^{\text{np}}(i)$ in the $(N,NK,M)$-non-private scheme is given by 
\begin{align*}
\bar{d}^{\text{np}}(i) = i \pmod N -c_k \text{ for } (k-1)N < i \leq kN. 
\end{align*}
The decodability follows since  $\bar{d}^{\text{np}}((k-1)N+s_k) $ is same as $\bar{d}(k)$.
We can observe that the rate of transmission in $(N,K,M)$-private scheme is equal to that of the $(N,NK,M)$-non-private scheme.
Next lemma shows that 
$(N,K,M)$-private scheme  also provides the demand privacy.
\begin{lemma}
	\label{Lem_Gen}
$(N,K,M)$-private scheme achieves the demand privacy defined in~\eqref{Eq_instant_priv}.
\end{lemma}
\begin{proof}
To prove this lemma, we  show that for $k \in [K]$,
\begin{align}
I(\ND{k}; Z_k,D_k,\TRS, \overline{W}) = 0,\label{Eq_priv_cond}
\end{align}
which implies~\eqref{Eq_instant_priv}.
For a given $\overline{W} = \overline{w}$,
let $X=(\bar{c},x^{\text{np}}_{\bar{c}}(\overline{w})),  \overline{D}=\overline{d}, S_k = s_k, Z_k = (c_k(s_k,\overline{w}),s_k)$ be such that $\text{Pr}(X=\left(\bar{c},x^{\text{np}}_{\bar{c}}(\overline{w})), Z_k = (c_k(s_k,\overline{w}),s_k), \overline{D}=\overline{d}, \overline{S}=\overline{s}|\overline{W} = \overline{w}\right)>0$. 
Then, we get
\begin{align}
\text{Pr}( \ND{k} & =\Nd{k}|Z_{k}=(c_k(s_k,\overline{w}),s_k),X=(\bar{c},x^{\text{np}}_{\bar{c}}(\overline{w})), D_k=d_k,S_k=s_k, \overline{W} = \overline{w}) \notag \\
&= \text{Pr}\left(\ND{k} = \Nd{k}| D_k=d_k,S_k=s_k, \bar{C}=\bar{c}, \overline{W} = \overline{w}\right)  \label{Eq_priv1}\\
&= \text{Pr}\left(\ND{k} = \Nd{k}| D_k=d_k,S_k=s_k, \bar{C}=\bar{c}\right)  \label{Eq_priv1a}\\
& = \text{Pr}\left(\NS{k} = \Ns{k}\right)\label{Eq_priv2}\\
&   =\left({\frac{1}{N}}\right)^{K-1}. \label{Eq_priv3}
\end{align}
Here, \eqref{Eq_priv1} follows because $x^{\text{np}}_{\bar{c}}(\bar{w})$ is a deterministic function of $\overline{w}$,  and $c_k(s_k,\bar{w})$ is a deterministic function of $\overline{w}$ and $s_k$. 
\eqref{Eq_priv1a} follows since $\overline{W}$ is independent of $\bar{C}, \bar{S}$ and $\overline{D}$.
 \eqref{Eq_priv2} follows due to the fact that 
$\bar{C}$ is a function of $\bar{S}$ and $\overline{D}$ as given in~\eqref{Def_C}.
\eqref{Eq_priv3} follows since $S_k \sim \text{Unif}\{[N]\} $ for all $k \in [K]$.
Thus, we have~\eqref{Eq_priv_cond} which proves the lemma.
\end{proof}
For the  $(N,NK,M)$-non-private scheme, 
 $\frac{{NK \choose KM+1 }-{NK-N \choose KM+1}}{{NK \choose KM}}$ is an  achievable rate for  $M \in \{0, 1/K,2/K,\ldots,N\}$. This shows the achievability of $\Rt$ in~\eqref{Eq_achvble_regn} for $M \in \{0, 1/K,2/K,\ldots,N\}$.

\subsection{Proof of Theorem~\ref{Thm_tightness}}
To prove the theorem, we first give some notations and inequalities.
First recall that  $\Rto $  and $\Rm$ denote the optimal rate with privacy and without privacy as defined in \eqref{Eq_opt_rate_priv} and  \eqref{Eq_opt_rate_nopriv}, respectively. $\Rt$ is the achievable rate given in Theorem~\ref{Thm_genach}. 
Let $R^{\text{YMA}}(N,K,M)$ denote the achievable rate with no privacy using the scheme given in \cite{Yu18} for $N$ files, $K$ users, and memory $M$. For parameter $r_2 = \frac{KM}{N}$, it is given by 
\begin{align}
 R^{\text{YMA}}(N,K,M) & = \frac{{K \choose r_2+1} - {K-\min(N,K) \choose r_2+1}}{{K \choose r_2}}, r_2 \in \{0, 1,\ldots, K\}.  \label{Eq_rate1}
\end{align}
Similarly, let $R^{\text{MN}}(N, K,M)$ denote an achievable rate with no privacy using the scheme in \cite{Maddah14}. For $M\in \{ 0, N/K, 2N/K, \ldots, N \}$, it is given by
\begin{align}
R^{\text{MN}}(N, K,M) &= K\left(1-\frac{M}{N}\right)\min\left(\frac{1}{1+\frac{KM}{N}},\frac{N}{K}\right). \label{Eq_rate2}
\end{align}

In Theorem~\ref{Thm_genach}, we showed that  $\Rt =R^{\text{YMA}}(N,NK,M)$. 
It was shown in~\cite{Yu19} (see Appendix~J) that the lower convex envelop of the points in~\eqref{Eq_rate1}  is same as the piecewise linear interpolation of adjacent  points.  So, $\Rtc$ defined in Definition~\ref{Def_convx} satisfies the following:
\begin{align*}
\Rtc = R^{\text{YMA}}(N,NK,M) \quad \forall M \geq 0.
\end{align*}
Then we have the following inequalities which hold for all $M\geq 0$:
\begin{align}
\Rm  \stackrel{(a)}{\leq} \Rto \stackrel{(b)}{\leq} \Rtc = R^{\text{YMA}}(N,NK,M)
\stackrel{(c)}{\leq} R^{\text{MN}}(N, NK,M), \label{Eq_rate_ineq1}
\end{align}
$(a)$ follows from the fact that the optimal rate required with demand privacy is  larger than that of without privacy, $(b)$ follows since an achievable rate is lower-bounded by the optimal rate, and the inequality $ R^{\text{YMA}}(N,NK,M) \leq R^{\text{MN}}(N, NK,M)$ in
$(c)$ is shown in \cite{Yu18}.



\underline{{\em Proof of part~\ref{Thm_tight_part1})}: ($N \leq K$)}

To prove  $\frac{\Rtc}{\Rto} \leq 8$, we show that if  $N \leq K$, then
\begin{align}
\frac{R^{\text{MN}}(N, NK,M)}{\Rm } \leq 8.  \label{Eq_ratio_bnd1}
\end{align}
Then the result follows from \eqref{Eq_rate_ineq1}. To obtain \eqref{Eq_ratio_bnd1},
we  first show that 
\begin{align}
\frac{R^{\text{MN}}(N, NK,M)}{R^{\text{MN}}(N, K,M)} \leq 2  \text{ for } M\in \{ 0, N/K, 2N/K, \ldots, N \}.\label{Eq_ratio_bnd2}
\end{align}
For $M\in \{ 0, N/K, 2N/K, \ldots, N \}$,
we have
\begin{align}
\frac{R^{\text{MN}}(N, NK,M)}{R^{\text{MN}}(N, K,M)} = \frac{N\min\left(\frac{1}{1+KM},\frac{1}{K}\right)}{\min\left(\frac{1}{1+\frac{KM}{N}},\frac{N}{K}\right)}. \label{ratio}
\end{align}

\noindent \underline {Case 1: $0 \leq M \leq 1-\frac{N}{K}$}

We first find $\min\left(\frac{1}{1+KM},\frac{1}{K}\right)$ and $\min\left(\frac{1}{1+\frac{KM}{N}},\frac{1}{K}\right) $. 
\begin{align*}
\frac{1}{1+KM} & \geq \frac{1}{1+K(1 - N/K)}\\
& = \frac{1}{K-N+1}\\
& > \frac{1}{K} \mbox{ for } N >1.
\end{align*}
So, $\min\left(\frac{1}{1+KM},\frac{1}{K}\right) = \frac{1}{K}$. Further,
\begin{align*}
\frac{1}{1+\frac{KM}{N}} &  \geq \frac{1}{1+\frac{K}{N}(1-N/K)}\\
& = \frac{N}{K}.
\end{align*}
So $\min\left(\frac{1}{1+\frac{KM}{N}},\frac{N}{K}\right) = \frac{N}{K}$. Hence~\eqref{ratio} gives 1.

%
%
%

\noindent \underline{Case 2: $1-\frac{N}{K} \leq M \leq 1-\frac{1}{K}$}

In this case, we get
\begin{align*}
\min\left(\frac{1}{1+KM},\frac{1}{K}\right) = \frac{1}{K} \mbox{ if } 1-\frac{N}{K} \leq M \leq 1-\frac{1}{K}
\end{align*}
 and
\begin{align*}
\min\left(\frac{1}{1+\frac{KM}{N}},\frac{N}{K}\right) =\frac{1}{1+\frac{KM}{N}} \; \mbox{ if }  1-\frac{N}{K} \leq M \leq 1-\frac{1}{K}.
\end{align*}
 Then from~\eqref{ratio},  it follows that
\begin{align*}
\frac{ R^{\text{MN}}(N, NK,M)}{ R^{\text{MN}}(N,K,M)}  & =\frac{N}{K}\left(1+\frac{KM}{N}\right)\\
&=\frac{N}{K}+M \; \mbox{ if } 1-\frac{N}{K} \leq M \leq 1-\frac{1}{K} \\
& \leq 2,
\end{align*}
where the last inequality follows since $\frac{N}{K} \leq 1$ and $M \leq 1$.

\noindent  \underline{Case 3 : $ 1-\frac{1}{K} \leq M \leq N$}

In this case, we obtain
\begin{align*}
\min\left(\frac{1}{1+KM},\frac{1}{K}\right) = \frac{1}{1+KM} \; \mbox{ if } 1-\frac{1}{K} \leq M \leq N
\end{align*}
and
\begin{align*}
\min\left(\frac{1}{1+\frac{KM}{N}},\frac{N}{K}\right) =\frac{1}{1+\frac{KM}{N}}  \; \mbox{ if } 1-\frac{1}{K} \leq M \leq N.
\end{align*}

Then from \eqref{ratio}, we get the following   
\begin{align}
\frac{ R^{\text{MN}}(N, NK,M)}{ R^{\text{MN}}(N,K,M)} & = \frac{N}{1+KM}\left(1+\frac{KM}{N}\right) \notag \\
&=\frac{N+KM}{1+KM}\notag \\
&=\frac{N-1}{1+KM}+1. \label{gap_temp2}  
\end{align}
Further,
\begin{align*}
M \geq 1-\frac{1}{K} & \implies KM \geq K-1,\\
& \implies KM \geq N -1 \; (\mbox{Since $K\geq N$}),\\
& \implies KM +1 \geq N-1,\\
& \implies \frac{N-1}{1+KM} \leq 1.
\end{align*} 
Then from \eqref{gap_temp2}, we get $\frac{ R^{\text{MN}}(N, NK,M)}{ R^{\text{MN}}(N, K,M)}  \leq 2$. Thus, we have \eqref{Eq_ratio_bnd2}.   


Let $R^{\text{MN}}_{\text{lin}}(N,K,M)$ the region obtained  by linearly interpolating the adjacent memory points given in~\eqref{Eq_rate2}. Then it follows from~\eqref{Eq_ratio_bnd2} that 
\begin{align}
\frac{ R^{\text{MN}}(N, NK,M)}{ R^{\text{MN}}_\text{lin}(N,K,M)}  \leq 2 \quad  \forall M\geq 0. \label{Eq_1ratio1}
\end{align}
We have the following lemma.
\begin{lemma}
	\label{Lem_bound_lin}
	For $N\leq K$, the following holds:
\begin{align}
 \frac{R^{\text{MN}}_\text{lin}(N, K,M) }{\Rm } \leq 4 \quad \forall M\geq 0. \label{Eq_1ratio2}
\end{align}
\end{lemma}

\begin{proof}
See Appendix~\ref{Sec_append}.
\end{proof}

From~\eqref{Eq_1ratio1} and ~\eqref{Eq_1ratio2}, we obtain~\eqref{Eq_ratio_bnd1}.
 This completes the proof of  part~\ref{Thm_tight_part1}).

\underline{{\em Proof of part~\ref{Thm_tight_part2})}: ($N \geq K$, $M\geq N/K$)}


We first show that 
\begin{align}
\frac{\Rt}{R^{\text{YMA}}(N, K,M)} \leq 2  \text{ for } M\in \{  N/K, 2N/K, \ldots, N \}.\label{Eq_ratio_prt2_bnd2}
\end{align}

Since $\min(N,K) = K$, we get for all $ M\in \{  N/K, 2N/K, \ldots, N \}$, i.e., $r_2= \frac{KM}{N} \in \{1, \ldots, K\}$, 
\begin{align}
R^{\text{YMA}}(N, K,M) & = \frac{{K \choose r_2+1} - {K-\min(N,K) \choose r_2+1}}{{K \choose r_2}} \notag   \\
& = \frac{{K \choose r_2+1}}{{K \choose r_2}}  \notag \\
&= \frac{K-r_2}{r_2+1}. \label{Eq_Rate_YMA}
\end{align}
 Since 
 $\Rt = R^{\text{YMA}}(N, NK,M)$, we get for all $ M\in \{  N/K, 2N/K, \ldots, N \}$, i.e., $r_1= MK  \in \{N, \ldots, NK\}$, 
 \begin{align}
\Rt & =\frac{{NK \choose r_1+1} - {NK-N \choose r_1+1}}{{NK \choose r_1}} \notag \\
& \leq \frac{{NK \choose r_1+1}}{{NK \choose r_1}} \notag \\
& = \frac{NK-r_1}{r_1+1}. \label{rate_32}
\end{align}
 By dividing \eqref{rate_32} by \eqref{Eq_Rate_YMA}, we get
\begin{align}
\frac{\Rt}{ R^{\text{YMA}}(N, K,M)} &\leq \frac{(NK-r_1)(r_2+1)}{(r_1+1)(K-r_2)}. \label{Eq_part2_inq}
\end{align}
For all $ M\in \{  N/K, 2N/K, \ldots, N \}$, it follows from \eqref{Eq_part2_inq} that 
\begin{align}
\frac{\Rt}{R^{\text{YMA}}(N, K,M)} &\leq \frac{N(r_2+1)}{Nr_2+1} \notag \\
&\leq 1+ \frac{N-1}{Nr_2+1}	 \notag\\
& \leq 2, \label{Eq_part2_ineq2}
\end{align}
using $r_2 = \frac{KM}{N} \geq 1$ for $M \geq N/K$.  Thus~\eqref{Eq_ratio_prt2_bnd2} follows. 


Since the lower convex envelop of the points $R^{\text{YMA}}(N,K,M), M\in \{ N/K, 2N/K, \ldots, N \}$ is same as the piecewise linear interpolation of adjacent  points~\cite{Yu18}, we obtain using~\eqref{Eq_part2_ineq2} 
\begin{align}
\frac{\Rtc}{R^{\text{YMA}}(N, K,M)} \leq 2  \quad \forall M \geq N/K. \label{Eq_YMA_bnd1}
\end{align}
Furthermore, it  was shown in~\cite{Yu18} that 
\begin{align}
\frac{R^{\text{YMA}}(N, K,M) }{\Rm } \leq 2  \quad \forall M \geq 0. \label{Eq_YMA_bnd2}
\end{align}
Then from~\eqref{Eq_YMA_bnd1} and~\eqref{Eq_YMA_bnd2}, we obtain
\begin{align*}
\frac{\Rtc}{\Rm} \leq 4 \quad \forall M \geq N/K.
\end{align*}
Since $ \Rm \leq \Rto $, it follows that
\begin{align*}
\frac{\Rtc}{\Rto} \leq 4 \quad \forall M \geq N/K.
\end{align*}

\underline{{\em Proof of part~\ref{Thm_tight_part3})}:}

 For $M=\frac{NK-1}{K}$ in $\Rt$ gives
\begin{align*}
R^p\left(N,K,\frac{NK-1}{K}\right) & = \frac{{NK \choose NK }-{NK-N \choose NK}}{{NK \choose NK-1}} \\
&= \frac{1}{{NK \choose NK-1}} \\
& = \frac{1}{NK}.
\end{align*}
For $M=N$, we have $R^p(N,K,0) = 0$. 
It was shown in~\cite[Theorem~2]{Maddah14} that  for $N$ files, $K$ users and memory $M$ with no privacy constraint, the following holds:
\begin{align}
\Rm \geq 1-\frac{M}{N}. \label{Eq_convrs_line}
\end{align}
It is easy to verify that the points $(\frac{NK-1}{K},\frac{1}{NK})$ and $(N,0)$ lie on the line given by~\eqref{Eq_convrs_line}. This shows that $\Rt = \Rto$ for $M\geq \frac{(NK-1)}{K}$.

\hfill{\rule{2.1mm}{2.1mm}}
\begin{remark}
		We note that for $N \leq K$, our scheme is order optimal within a factor of 4 for the memory regime $M \in [0,1-N/K] \cup [\frac{N}{2},N]$ which can be seen from the proof of Theorem~\ref{Thm_tightness} part~\ref{Thm_tight_part1}).
\end{remark}

\section{Discussion and comparison}
\label{sec_discuss}

In this paper, we studied the case where the number of requests from a user, denoted by $L$, is always one. Without the privacy constraints, this  setup was studied in~\cite{Maddah14}. The recent work~\cite{Wan19} also studied the case of $L>1$ which was introduced  in~\cite{Ji15} under no privacy constraint. In the following, we compare our scheme with that of~\cite{Wan19} for $L=1$.

\begin{enumerate}
	\item For $N \leq K$, we compare the performance of our scheme with the general scheme given in~\cite{Wan19}. Simulation results show that our scheme performs strictly better if $M< N/2$ compared to the  scheme which achieves $R_{\text{com}}$ in~\cite{Wan19}. For example, for $N=5$ and $K=10$, the memory-rate trade-off is plotted for both the schemes in Fig~\ref{fig_comparison_privacy1}. As we can observe our scheme  performs better when $M < N/2$ and it is also order optimal from the converse for all the memory regions. For $M \geq N/2$, we observed that for some memory regimes our schemes performs better while some memory regime, the scheme of~\cite{Wan19} performs better.
	\item An improved scheme was provided in~\cite[Theorem~7]{Wan19}, for a special case where all users have distinct demands which inherently assumes that $N \geq K$.
	It is shown that the achievable rates with the improved scheme $R_{d,new}$ is order optimal within a factor of 4 when $M \geq N/K$. However, in Thoerem~\ref{Thm_tightness} part~\ref{Thm_tight_part2}), we show that our scheme is order optimal within a factor of 4 without the assumption on the demand vector for $N \geq K$ and $M \geq N/K$. In Fig~\ref{fig_comparison_privacy2}, the performance of different schemes are plotted for $N=20$ and $K=10$. 
\end{enumerate}

\begin{figure}[h] 
	\centering
	\includegraphics[scale=0.4]{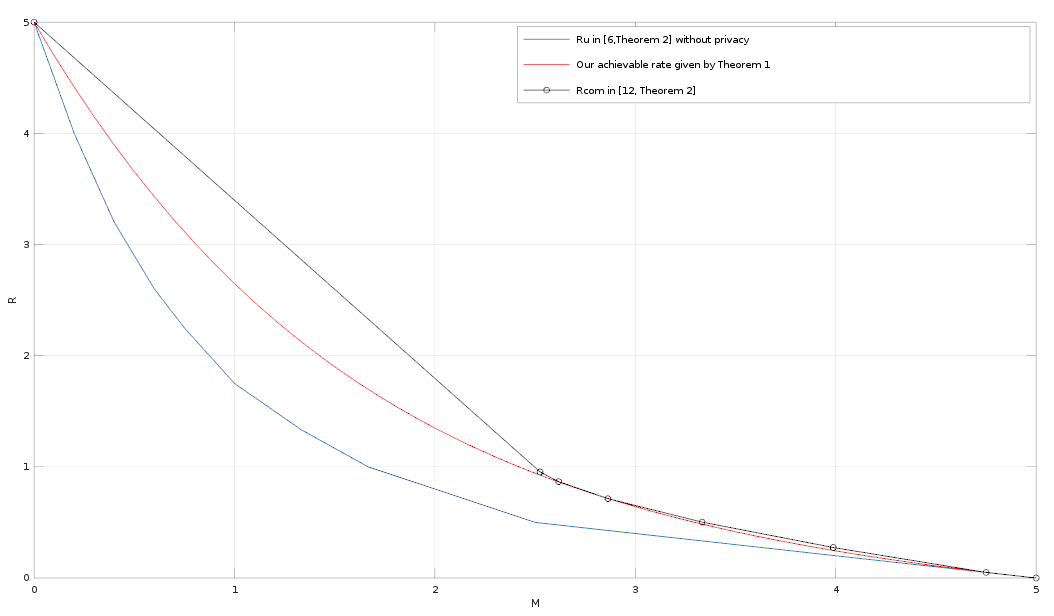}
	\caption{Comparison of the schemes for  N=5, K=10}
	\label{fig_comparison_privacy1}
\end{figure}

\begin{figure}[h] 
	\centering
	\includegraphics[scale=0.4]{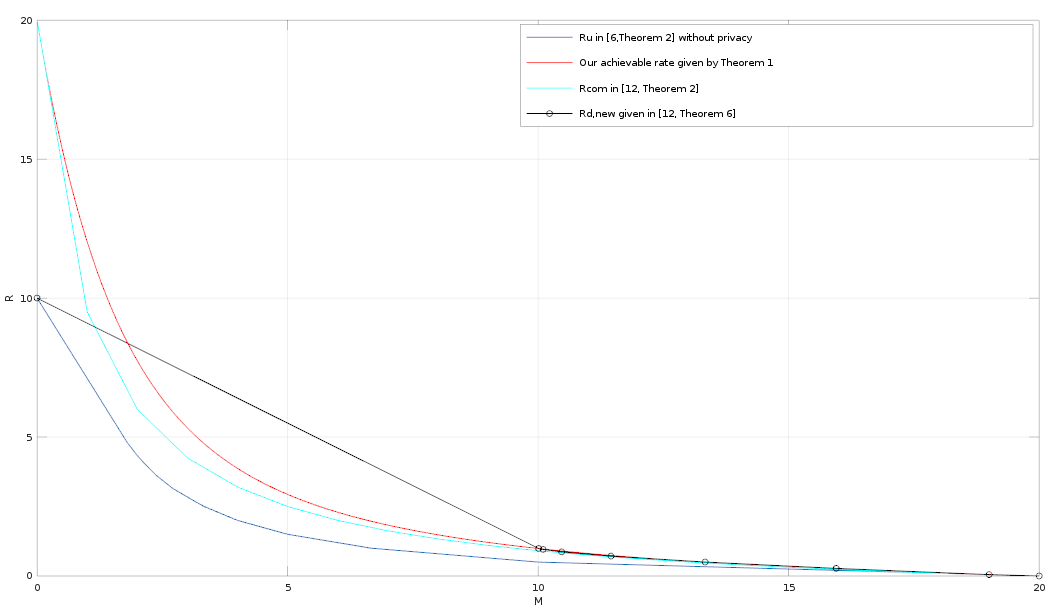}
	\caption{Comparison of the schemes  for N=20, K=10}
	\label{fig_comparison_privacy2}
\end{figure}

\begin{appendices}
\section{Proof of Lemma~\ref{Lem_bound_lin}}
\label{Sec_append}
To prove the lemma, we use the converse results given in~\cite[Theorem~2]{Ghasemi17} for coded caching under no privacy. We also use the following result.

\begin{claim}
	\begin{enumerate}
		\item For $N \leq K, \forall M \geq 0$,   $R^{\text{MN}}_{\text{lin}}(N,K,M)$  is monotonically non-increasing.
		\item 	For $N \leq K, 1-N/K \leq M \leq N$, $R^{\text{MN}}_{\text{lin}}(N,K,M)$ is convex.
	\end{enumerate}
	
	\begin{proof}
		We first show part~1).
		For $0 \leq M \leq 1-N/K$,  as we have shown in Case~1  that 
		\begin{align*}
		R^{\text{MN}}_{\text{lin}}(N,K,M) & = N(1-M/N). 
		\end{align*}
		Here it is easy to see that $	R^{\text{MN}}_{\text{lin}}(N,K,M) $ is monotonically non-increasing. 
		Let us consider the regime $ 1-N/K \leq M \leq N$.
		For all $t_0 = \frac{KM}{N} \in\{0,1, \ldots, K\}$ such that $\frac{t_0N}{K} \geq 1-N/K$, we have
		\begin{align}
		R^{\text{MN}}_{\text{lin}}\left(N,K,\frac{t_0N}{K}\right) & = K(1-M/N)\frac{1}{1+\frac{KM}{N}}. 
		\end{align}
		By differentiating  $K(1-M/N)\frac{1}{1+\frac{KM}{N}}$  w.r.t. $M$, we get
		\begin{align*}
		\frac{d}{dM}\left(K\frac{\left(1-\frac{M}{N}\right)}{1+\frac{KM}{N}}\right)
		&=-\frac{K(K+1)}{N}\left(1+\frac{KM}{N}\right)^{-2}\\
		& \leq 0
		\end{align*}
		which shows that 	$R^{\text{MN}}_{\text{lin}}(N,K,M) $ is monotonically non-increasing in $ 1-N/K \leq M \leq N$. To prove part~2), we take the second derivative of $N(1-M/N)\frac{1}{1+\frac{KM}{N}}$ w.r.t. $M$:
		\begin{align*}
		&\frac{d}{dM}\left(-\frac{K(K+1)}{N}\left(1+\frac{KM}{N}\right)^{-2}\right)=2\frac{K^2(K+1)}{N^2}\left(1+\frac{KM}{N}\right)^{-3}
		\end{align*}	
		which proves part~2).
	\end{proof}
\end{claim}	

Now we consider the three memory regions studied in the proof of~\cite[Theorem~2]{Ghasemi17}. For $N \leq K$, the three regions are as follows: 

\underline{Region I: $0 \leq M \leq 1 $: }

Since $R^{\text{MN}}_\text{lin}(N, K,0) = N$, and also that $R^{\text{MN}}_\text{lin}(N, K,M) $ is monotonically non-increasing in $M$, we get
\begin{align*}
R^{\text{MN}}_\text{lin}(N, K,M) \leq N.
\end{align*}

In this regime, it was shown~\cite[Theorem~2]{Ghasemi17} that 
$\Rm \geq N/4$. Then, we have 
\begin{align*}
\frac{R^{\text{MN}}_\text{lin}(N, K,M) }{\Rm } \leq 4 \quad \mbox{ for } 0 \leq M  \leq 1.
\end{align*}

\underline{Region II: $1 \leq M \leq N/2 $: }

Let  $f_1(M): = \frac{N}{M} - \frac{1}{2}$. 
For $t_0 = \lfloor{\frac{KM}{N}}\rfloor$ and $\frac{Nt_0}{K} \leq M \leq \frac{N(t_0+1)}{K}$, 
it was shown~\cite[Theorem~2]{Ghasemi17}  that 
\begin{align}
R^{\text{MN}}_{\text{c}}(N,K,M) \leq R^{\text{MN}}\left(N,K,\frac{Nt_0}{K}\right) \leq f_1(M). \label{Eq_2regio1}
\end{align}

Since $R^{\text{MN}}_{\text{lin}}(N,K,\frac{Nt_0}{K}) = R^{\text{MN}}(N,K,\frac{Nt_0}{K})$ and also that
$R^{\text{MN}}_{\text{lin}}(N,K,M)$ is non-increasing, we get
\begin{align}
R^{\text{MN}}_\text{lin}(N, K,M) \leq  R^{\text{MN}}\left(N,K,\frac{Nt_0}{K}\right). \label{Eq_2regio2}
\end{align}
It was also shown~\cite[Theorem~2]{Ghasemi17}  that 
\begin{align}
\frac{f_1(M)}{\Rm} \leq 4. \label{Eq_2regio3}
\end{align}
From~\eqref{Eq_2regio1}~\eqref{Eq_2regio2} and~\eqref{Eq_2regio3}, it follows that 
\begin{align}
\frac{R^{\text{MN}}_\text{lin}(N, K,M) }{\Rm } \leq 4.
\end{align}

\underline{Region III: $N/2 \leq M \leq N $: }

Let $f_2(M): = 2(1-M/N)$. 
For $t_0 = \lfloor{\frac{K}{2}}\rfloor$ and for all $M \geq \frac{N(t_0)}{K}$, 
it was shown~\cite[Theorem~2]{Ghasemi17}  that for $\lambda = \frac{1-M/N}{1-t_0/K}$,
\begin{align}
R^{\text{MN}}_{\text{c}}(N,K,M) \leq \lambda R^{\text{MN}}\left(N,K,\frac{Nt_0}{K}\right) \leq f_2(M), \label{Eq_3regio1}
\end{align}
and also that
\begin{align}
\frac{f_2(M)}{\Rm} \leq 2. \label{Eq_3regio2}
\end{align}
By definition of $R^{\text{MN}}_{\text{lin}}$, we have,
\begin{align}
R^{\text{MN}}_{\text{lin}}\left(N,K,\frac{Nt_0}{K}\right) & = R^{\text{MN}}\left(N,K,\frac{Nt_0}{K}\right).\label{small_eq1}
\end{align}
We also have
\begin{align}
R^{\text{MN}}_{\text{lin}}(N,K,N) & = R^{\text{MN}}(N,K,N) \notag \\
& = 0. \label{small_eq2}
\end{align}
Then from the convexity of $R^{\text{MN}}_{\text{lin}}(N,K,M) $ in this regime, it follows that 
\begin{align}
R^{\text{MN}}_\text{lin}(N, K,M) \leq \lambda  R^{\text{MN}}_\text{lin}\left(N,K,\frac{Nt_0}{K}\right)+ (1-\lambda)R^{\text{MN}}_{\text{lin}}(N,K,N) \label{small_eq3}
\end{align}
From ~\eqref{small_eq1}~\eqref{small_eq2} and~\eqref{small_eq3}, we get
\begin{align}
R^{\text{MN}}_\text{lin}(N, K,M) \leq \lambda  R^{\text{MN}}\left(N,K,\frac{Nt_0}{K}\right).\label{Eq_3regio3}
\end{align}
From~\eqref{Eq_3regio1}~\eqref{Eq_3regio2} and~\eqref{Eq_3regio3}, we obtain
\begin{align}
\frac{R^{\text{MN}}_\text{lin}(N, K,M) }{\Rm } \leq 2.
\end{align}
This completes the proof of the lemma.

\end{appendices}

 \bibliographystyle{IEEEtran}
 \bibliography{Bibliography.bib}

\begin{thebibliography}{10}
\providecommand{\url}[1]{#1}
\csname url@samestyle\endcsname
\providecommand{\newblock}{\relax}
\providecommand{\bibinfo}[2]{#2}
\providecommand{\BIBentrySTDinterwordspacing}{\spaceskip=0pt\relax}
\providecommand{\BIBentryALTinterwordstretchfactor}{4}
\providecommand{\BIBentryALTinterwordspacing}{\spaceskip=\fontdimen2\font plus
\BIBentryALTinterwordstretchfactor\fontdimen3\font minus
  \fontdimen4\font\relax}
\providecommand{\BIBforeignlanguage}[2]{{%
\expandafter\ifx\csname l@#1\endcsname\relax
\typeout{** WARNING: IEEEtran.bst: No hyphenation pattern has been}%
\typeout{** loaded for the language `#1'. Using the pattern for}%
\typeout{** the default language instead.}%
\else
\language=\csname l@#1\endcsname
\fi
#2}}
\providecommand{\BIBdecl}{\relax}
\BIBdecl

\bibitem{Maddah16}
M.~A. Maddah-Ali and U.~Niesen, ``Coding for caching: fundamental limits and
  practical challenges,'' \emph{IEEE Communications Magazine}, vol.~54, no.~8,
  pp. 23--29, August 2016.

\bibitem{Maddah14}
------, ``Fundamental limits of caching,'' \emph{IEEE Transactions on
  Information Theory}, vol.~60, no.~5, pp. 2856--2867, May 2014.

\bibitem{Amiri17}
M.~{Mohammadi Amiri} and D.~{Gunduz}, ``Fundamental limits of coded caching:
  Improved delivery rate-cache capacity tradeoff,'' \emph{IEEE Transactions on
  Communications}, vol.~65, no.~2, pp. 806--815, Feb 2017.

\bibitem{Zhang18}
K.~{Zhang} and C.~{Tian}, ``Fundamental limits of coded caching: From uncoded
  prefetching to coded prefetching,'' \emph{IEEE Journal on Selected Areas in
  Communications}, vol.~36, no.~6, pp. 1153--1164, June 2018.

\bibitem{Vilardebo18}
J.~{G\'omez-Vilardeb\'o}, ``Fundamental limits of caching: Improved rate-memory
  tradeoff with coded prefetching,'' \emph{IEEE Transactions on
  Communications}, vol.~66, no.~10, pp. 4488--4497, Oct 2018.

\bibitem{Yu18}
Q.~Yu, M.~A. Maddah-Ali, and A.~S. Avestimehr, ``The exact rate-memory tradeoff
  for caching with uncoded prefetching,'' \emph{IEEE Transactions on
  Information Theory}, vol.~64, no.~2, pp. 1281--1296, Feb 2018.

\bibitem{Ghasemi17}
H.~Ghasemi and A.~Ramamoorthy, ``Improved lower bounds for coded caching,''
  \emph{IEEE Transactions on Information Theory}, vol.~63, no.~7, pp.
  4388--4413, July 2017.

\bibitem{Wang18}
C.~{Wang}, S.~{Saeedi Bidokhti}, and M.~{Wigger}, ``Improved converses and gap
  results for coded caching,'' \emph{IEEE Transactions on Information Theory},
  vol.~64, no.~11, pp. 7051--7062, Nov 2018.

\bibitem{Wan16}
{Kai Wan}, D.~{Tuninetti}, and P.~{Piantanida}, ``On the optimality of uncoded
  cache placement,'' in \emph{2016 IEEE Information Theory Workshop (ITW)},
  Sep. 2016, pp. 161--165.

\bibitem{Sengupta15}
A.~Sengupta, R.~Tandon, and T.~C. Clancy, ``Fundamental limits of caching with
  secure delivery,'' \emph{IEEE Transactions on Information Forensics and
  Security}, vol.~10, no.~2, pp. 355--370, Feb 2015.

\bibitem{Ravindrakumar18}
V.~Ravindrakumar, P.~Panda, N.~Karamchandani, and V.~M. Prabhakaran, ``Private
  coded caching,'' \emph{IEEE Transactions on Information Forensics and
  Security}, vol.~13, no.~3, pp. 685--694, March 2018.

\bibitem{Wan19}
K.~{Wan} and G.~{Caire}, ``On coded caching with private demands,'' {\tt
  arXiv:1908.10821 [cs.IT]}, Aug 2019.

\bibitem{Yu19}
Q.~Yu, M.~A. Maddah-Ali, and A.~S. Avestimehr, ``Characterizing the rate-memory
  tradeoff in cache networks within a factor of 2,'' \emph{IEEE Transactions on
  Information Theory}, vol.~65, no.~1, pp. 647--663, Jan 2019.

\bibitem{Ji15}
M.~{Ji}, A.~{Tulino}, J.~{Llorca}, and G.~{Caire}, ``Caching-aided coded
  multicasting with multiple random requests,'' in \emph{2015 IEEE Information
  Theory Workshop (ITW)}, April 2015, pp. 1--5.

\end{thebibliography}





\end{document}